\documentclass[10pt]{article}
\usepackage{epsfig,graphics,graphicx,subfigure,color}
\usepackage{amsfonts,amssymb,amsthm,amsmath,amsbsy,mathtools,cuted,kotex}

\newtheorem{theorem}{Theorem}
\newtheorem{corollary}{Corollary}
\newtheorem{lem}[theorem]{Lemma}
\def\Bbb R{{\rm \bf R}}

\def\gathered{\begin{array}{c}}
\def\endgathered{\end{array}}
\def\text{\mbox}

\begin{document}

\title {Three range measurements with multiplicative noises for single source localization problem}
\author{Kiwoon KWON\footnote{ Department of Mathematics, Dongguk University-Seoul, 100715 Seoul, South Korea.
e-mail address:
kwkwon@dongguk.edu}
}
\maketitle

\begin{abstract}
This purpose of this paper is to locate a single localized source from three range measurements with multiplicative noises.
Although some minimization approaches for additive noise have been found, studies on the existence of solutions are rare.
We analyzed a situation with one or two solutions for the same multiplicative noise at three measurement sensors.  
A strategy for finding the best localized source when there are no solutions for the same multilicative noise is suggested that involves adjusting the multiplicative noise ratio.  The numerical simulation is conducted for three randomly generated measurement locations and their distances to the source. 
\noindent

\end{abstract}

\section{Introduction}
We consider the problem of locating a single radiating source from multiplicaive noisy range measurements collected using a network of passive sensors. This problem is relavant to Global Positioning Systems(GPS), wireless communications, surveillance etc. \cite{BS}. Most of the applications use many measurements, and only three measurements are used in the inverse scattering problem for phaseless far field patterns \cite{XXB} and the references therein.  Only three measurements cases will be examined in this paper.

Consider an array of $m$ sensors, and let $z_j\in \mathbb R^n$ denote the coordinate of the $j$-th sensor for $j=1,2,3$. Let $x\in \mathbb R^n$ denote the unknown source's coordinate vector, and let $d_j^0>0$ be a noisy observation of the range between the source and the $j$-th sensor:
\begin{equation}\label{eq:main}
                         d_j^0 = ||x-z_j||(1+\epsilon_j) , j=1,2,3.
 \end{equation}
This formulation differs from traditional additive noise formulation in \cite{ BTC,BP,BS}. If $n\ge 3$ and the source is not located on the plane defined by $z_1,z_2,$ and $z_3$, we can only approximate $x$ as the orthogonal projection to this plane. Thus, let us assume in this paper that the source location is on the two-dimensional plane defined by $z_1,z_2,$ and $z_3$. Without a loss of generality, let us assume  $d_3^0\ge d_2^0 \ge d_1^0$.  

Let us consider the multiplicaitve noise $\epsilon$. If $\epsilon$ is known exactly, $x$ can be defined without difficulty. However, it is  generally not exactly known due to the possible ill-posed nature of  the measurements. Let us assume that the devices and surrounding envirionments for the three measurement sensors are similar; that is to say, that $\epsilon$ is the solution to the following minimization problem:
\begin{equation}\label{eq:same}
 \epsilon = argmin_{(\epsilon_1,\epsilon_2,\epsilon_3)}[|\epsilon_3-\epsilon_2|+|\epsilon_3-\epsilon_1|].
 \end{equation} 
If there is a corresponding source $x$, the best candidate for $\epsilon$ is $ \epsilon_1 = \epsilon_2 = \epsilon_3$. In general, various minimization techniques can be used for these problems. However, in this paper, we will identify the exact source location $x$ and related noise $\epsilon$ when the solution exists. In addition, we proposed an approximation strategy for finding the exact source location $x$ by controlling $\epsilon$  to minimize the objective function in (\ref{eq:same}). 

In section 2, the case in which there exists a source location $x$ satisfying (\ref{eq:main}) when $\epsilon_1=\epsilon_2=\epsilon_3$ will be discussed. In section 3, an approximation strategy for when there does not exist $x$ satisfying (\ref{eq:main}) is proposed. Theorems supporting the strategies are also stated and proved, and numerical examples for all cases controlling $\epsilon$ are explained. The error for finding $x$ for randomly generated measurement data added by multiplicative noise is also investigated with increasing  multiplicative noise.

\section{The existence of source}
Let us denote the followings for $i,j=1,2,3$:
\begin{enumerate}
\item{$d_j= \frac{d^0_j}{1+\epsilon_j}$}
\item{$z_{ij} = z_i - z_j,$}
\item{$d_{ij}= ||z_{ij}||$}
\item{$b_{ij}^{in} = \frac{d_j}{d_j + d_i} z_i + \frac{d_i}{d_i+d_j} z_j$}
\item{$b_{ij}^{out} = \frac{d_j}{d_j - d_i} z_i + \frac{d_i}{d_i-d_j} z_j$ only when $d_i \neq d_j$}
\item{$b_{ij} = \frac{d_j^2}{d_j^2 - d_i^2} z_i + \frac{d_i^2}{d_i^2-d_j^2} z_j$ only when $d_i \neq d_j$}
\item{$r_{ij} = \frac{d_i d_j d_{ij}}{|d_i^2 - d_j^2|} $ only when $d_i \neq d_j$}
\end{enumerate}
Here, $b_{ij}^{in}$ and $b_{ij}^{out}$ are respectively the internally and externally dividing points with ratios $d_i$ and $d_j$ between two points $z_i$ and $z_j$. And $b_{ij}$ and $r_{ij}$ are the center and radius, respectively, of  the correspoding Apollonius circle. 
Without a loss of generality, let us assume 
$$   d_3 \ge d_2 \ge d_1 >0. $$

\begin{theorem}\label{th:th1}
Suppose that $\epsilon_1=\epsilon_2=\epsilon_3$. 
Then, there is a unique solution for (\ref{eq:main}) if and only if one of the following conditions hold:
 \begin{eqnarray}\label{eq:th1abc}
d_3 = d_2 = d_1 &:&  z_1, z_2, z_3 \mbox{ are not colinear, } \label{eq:th1a}\\
 d_3 = d_2 > d_1&:&    \frac{|z_{23} \cdot (b_{13}-b_{23}^{in})|}{d_{23}}   =   r_{13},\label{eq:th1b}\\
d_3 > d_2 \ge d_1 &:&   |b_{23}-b_{13}| = r_{23} + r_{13} \mbox{ or } |b_{23}-b_{13}| =  |r_{23}-r_{13}|. \label{eq:th1c}
\end{eqnarray}
And there are exactly two solutions  for (\ref{eq:main}) if and only if one of the following conditions hold
 \begin{eqnarray}\label{eq:th1de}
 d_3 = d_2 > d_1&:&    \frac{|z_{23} \cdot (b_{13}-b_{23}^{in})|}{d_{23}}   <   r_{13}\label{eq:th1d}\\
d_3 > d_2 \ge d_1 &:&    |r_{23}-r_{13}| <  |b_{23}-b_{13}|  < r_{23} + r_{13}\label{eq:th1e}
\end{eqnarray}
And there is no solution for (\ref{eq:main}) if and only if one of the following conditions hold
 \begin{eqnarray}\label{eq:th1fg}
 d_3 = d_2 > d_1&:&    \frac{|z_{23} \cdot (b_{13}-b_{23}^{in})|}{d_{23}}   >   r_{13},\label{eq:th1f}\\
d_3 > d_2 \ge d_1 &:&   |b_{23}-b_{13}|<  |r_{23}-r_{13}| ,\label{eq:th1g}\\
d_3 > d_2 \ge d_1 &:&   |b_{23}-b_{13}| >  r_{23} + r_{13}.\label{eq:th1h}
\end{eqnarray}
 \end{theorem}
 \begin{proof}
 If $z_1,z_2,$ and $z_3$ are not colinear, they make a triangle. Further, $d_3=d_2=d_1$ means that the solution is located at 
 the intersection of three lines that are perpendicular to each side passing through the midpoint of the side. The existence of 
 the solution is exactly the circumcenter. If $z_1,z_2,$ and $z_3$ are colinear and $d_3=d_2=d_1$, it can be easily verified that there is no solution to (\ref{eq:main}). This proves that the noncollinearity is an  equivalent condition of the  unique solution of (\ref{eq:main}) when $d_3=d_2=d_1$.\\
 If $d_3 = d_2 > d_1$, then the solution lies at the perpendicular line bisecting the side between $z_2$ and $z_3$
 and the Apollonius circle with radius $r_{13}$ centered at $b_{13}$. The distance from point $b_{13}$ 
 to the perpendicular line bisecting the side $\overline{z_2 z_3}$  is 
 $$   D =   \frac{|z_{23} \cdot (b_{13}-b_{23}^{in})|}{d_{23}}. $$
 Thus, the existence of the solution for (\ref{eq:main}) is equivalent to $D\le r_{13}$. In addition, there is a unique solution if and only if $D=r_{13}$, while there exist exactly two solutions if and only if $D<r_{13}$.\\
 If $d_3 > d_2 \ge d_1$, the existence of the solution for (\ref{eq:main}) is equivalent to the existence of the meeting points of two Apollonius's circles for the two sides $\overline{z_2 z_3}$ and $\overline{z_1 z_3}$.  If   $|b_{23}-b_{13}| = r_{23} + r_{13}$, the two circles circumscribe and if $|b_{23}-b_{13}| = |r_{23} - r_{13}|$, the two circles inscribe.  If   $|r_{23}-r_{13}| <  |b_{23}-b_{13}|  < r_{23} + r_{13}$, the two circles meet at two points. These prove conditions (\ref{eq:th1c}), (\ref{eq:th1e}), (\ref{eq:th1g}), and (\ref{eq:th1h}) . 
  \end{proof}

\begin{corollary}\label{co:th1}
If $\epsilon$ is exactly known, one of the following conditions hold:
 \begin{eqnarray}\label{eq:coth1}
d_3 = d_2 = d_1 &:&  z_1, z_2, z_3 \mbox{ are not colinear. } \label{eq:coth1a}\\
 d_3 = d_2 > d_1&:&    \frac{|z_{23} \cdot (b_{13}-b_{23}^{in})|}{d_{23}}   \le   r_{13}\label{eq:coth1b}\\
d_3 > d_2 \ge d_1 &:&    |r_{23}-r_{13}| \le  |b_{23}-b_{13}|  \le r_{23} + r_{13}.\label{eq:coth1c}
\end{eqnarray}
 \end{corollary}
\section{Approximation strategy}
Given $z_j, d_j, j=1,2,3$ and $\epsilon_1=\epsilon_2=\epsilon_3$, if one of conditions (\ref{eq:th1a}), (\ref{eq:th1b}), and (\ref{eq:th1c}) hold, then we can find the unique solution $x$ for (\ref{eq:main}) using Theorem \ref{th:th1}. If one of the two conditions (\ref{eq:th1d}) and (\ref{eq:th1e}) hold, there are two solutions for (\ref{eq:main}); let these two solutions be $x_1$ and $x_2$. Then we can choose 
$$ x=argmin_{x_i, i=1,2} \sum_{j=1,2,3} | ||z_j-x_i||(1+\epsilon_j)-d_j^0| 
. $$

If one of the three conditions (\ref{eq:th1f}), (\ref{eq:th1g}), and (\ref{eq:th1h}) hold, the assumption $\epsilon_1=\epsilon_2=\epsilon_3$ no longer holds, and we should change $\epsilon$ to have a solution. We will find a solution $x$ along with $\epsilon$ of the minimization problem such as 
$$    \min_{\epsilon=\epsilon(x)} [|\epsilon_3-\epsilon_2|+|\epsilon_3-\epsilon_1|]= \min_{x=x(\epsilon)}\left[ \left|\frac{d_3}{|z_3-x|} -   \frac{d_2}{|z_2-x|}\right| +  \left| \frac{d_3}{|z_3-x|} -   \frac{d_1}{|z_1-x|}\right|\right].$$ 
Rather than attempting to directly solve this minimization problem using known minimization methods, we will try another strategy to obtain a solution by controlling $\epsilon$, which results in $d$ changing accordingly.
\begin{itemize}
\item{ The order $d_3\ge d_2 \ge d_1$ should be fixed. }
\item{ The value $d_3$ is fixed and $d_2$ and/or $d_1$  increase(s).}
\item{ The values by which $d_2$ and $d_1$ increase should be as small as possible.}
\end{itemize}
Bearing in mind that $d_2$ and $d_1$ increase into $k_2d_2$ and $k_1d_1$ for $1\le  k_2 \le \frac{d_3}{d_2}$ and $1\le k_1\le \frac{k_2d_2}{d_1} $, let us define the followings: 
\begin{enumerate}
\item{$b_{ij}^{in} (k_i,k_j)= \frac{k_jd_j}{k_jd_j + k_id_i} z_i + \frac{k_id_i}{k_id_i+k_jd_j} z_j$}
\item{$b_{ij}^{out}(k_i,k_j) = \frac{k_jd_j}{k_jd_j - k_id_i} z_i + \frac{k_id_i}{k_id_i-k_jd_j} z_j$ only when $k_id_i \neq k_jd_j$}
\item{$b_{ij}(k_i,k_j) = \frac{k_j^2d_j^2}{k_j^2d_j^2 -k_i^2 d_i^2} z_i + \frac{k_i^2d_i^2}{k_i^2d_i^2-k_j^2d_j^2} z_j$ 
only when $k_id_i \neq k_jd_j$}
\item{$r_{ij} (k_i,k_j)= \frac{k_ik_jd_i d_j d_{ij}}{|k_i^2d_i^2 - k_j^2d_j^2|} $ only when $k_id_i \neq k_jd_j$}
\end{enumerate}

If condition (\ref{eq:th1f}) holds, there exists no solution if $\epsilon_1 = \epsilon_2 = \epsilon_3$. By the above strategy, we should increase $d_1$ upto $d_2$. Let 
$$D(k) := |z_{23} \cdot (b_{13}(k,1)-b_{23}^{in})|/d_{23}   -  r_{13}(k,1). $$ 
 Then, (\ref{eq:th1f})  implies that  $D(1)>0$. If there is a $k$ such that $1<k\le \frac{d_2}{d_1}$ and $D(k)\le 0$, then there is  a solution $x$ for (\ref{eq:main}) with $d_3=d_2>kd_1$. Let  $k_D = d_2/d_1$, 
\begin{theorem}\label{th:th2}
Suppose $ d_3 = d_2 > d_1$ and $ D(1)>0 $.  Then, $D(k)$ is a decreasing function for $k\in[1,k_D)$ and there is a $k\in[1,k_D)$ such that $D(k)=0$ and $d_3 = d_2 > kd_1$. 
\end{theorem}
\begin{proof}
 Let  $1\le k_1 < k_2 < k_D$ and  let the blue and red circles in Fig. \ref{fig:Dk}(a) respectively correspond to the Apollonius circles for $k_1$ and $k_2$. Let $\theta$ be the angle $\angle z_1 z_3 z_2$ and assume that this angle is acute. Even though Fig. \ref{fig:Dk}(a) is for acute angle $\theta$, the following points are also true for an obtuse angle by using $|\cos\theta|$.  Under these assumptions, we have 
\begin{eqnarray}
D(k_2) - D(k_1) &=& [|DG|-|DA|]-[|CF|-|CB|]
\nonumber\\
                               &=& [|DG|-|CF|] - [|DA|-|CB|]
 \nonumber\\                              
                               &=&  |CD||\cos\theta| - [|AB| + |CD|] 
 \nonumber\\
                                &=&  |CD|(|\cos\theta|-1) - |AB| < 0.
\end{eqnarray}
Therefore, $D(k)$ is a decreasing function on $[1,k_D)$.
\\
Note that
$$b_{13}(k,1) = \frac{d_3^2}{d_3^2 -k^2 d_1^2} z_1 + \frac{k^2d_1^2}{k^2 d_1^2 - d_3^2} z_3 =  \frac 1{1 -k^2 (d_1/d_3)^2}  z_{13}    +    z_3$$
$$|b_{13}(k,1) - \frac 1 2(z_1+z_3)| 
 =   \frac {1+k^2(d_1/d_3)^2}{2(1 -k^2 (d_1/d_3)^2)}  d_{13} $$
and
$$ r_{13}(k,1) = \frac{k (d_1/d_3) d_{13}}{1 - k^2( d_1/d_3)^2}  = |b_{13}(k,1) - b_{13}^{in}(k,1)|.$$
Since $b_{13}(k,1), b_{13}^{in}(k,1),$ and $\frac 1 2(z_1+z_3)$ lie on the same line, we have
$$\left| \frac{z_1+z_3} 2  - b_{13}^{in}(k,1)\right| =  \left| b_{13}(k,1) - \frac {z_1+z_3} 2\right| - |b_{13}(k,1) - b_{13}^{in}(k,1)| = \frac {1-k (d_1/d_3)}{1+k(d_1/d_3)}d_{13}. $$
Therefore $k$ is less than and sufficiently close to $k_D$ and $b_{13}^{in}(k,1)$ approaches $\frac{z_1+z_3} 2$; that is to say, $b_{13}^{in}(k_D-0,1)=\frac{z_1+z_3}2$. 
 
Thus, the Apollonius circle on the side $\overline{z_1 z_3}$ goes to the perpendicular line to $\overline{z_1 z_3}$ at $\frac{z_1+z_3}2$ as $k$ goes to $k_D-0$. Since the two perpendicular lines with respect to sides $\overline{z_2 z_3}$ and $\overline{z_1 z_3}$ meets, there is an Apollonius circle depending on $k\in[1,k_D)$ along $\overline{z_1 z_3}$ meeting the perpendicular line with respect to $\overline{z_2 z_3}$ as shown in Fig. \ref{fig:Dk}(b). 
\end{proof}
\begin{figure}
\begin{minipage}[t]{8cm}
\centerline{\epsfig{file=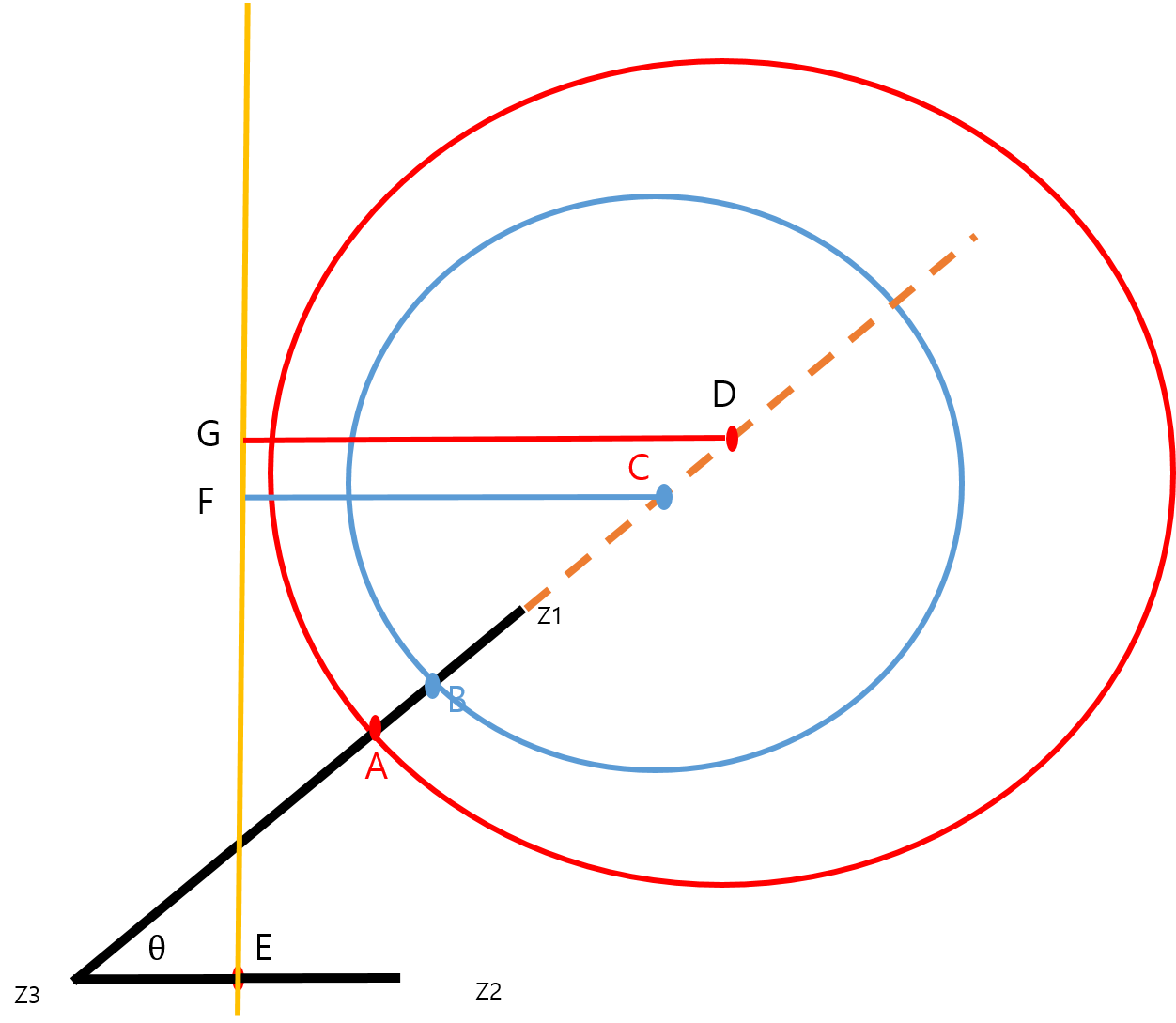, height=4cm,width=4cm,clip=2cm}}
\end{minipage}
\begin{minipage}[t]{8cm}
\centerline{\epsfig{file=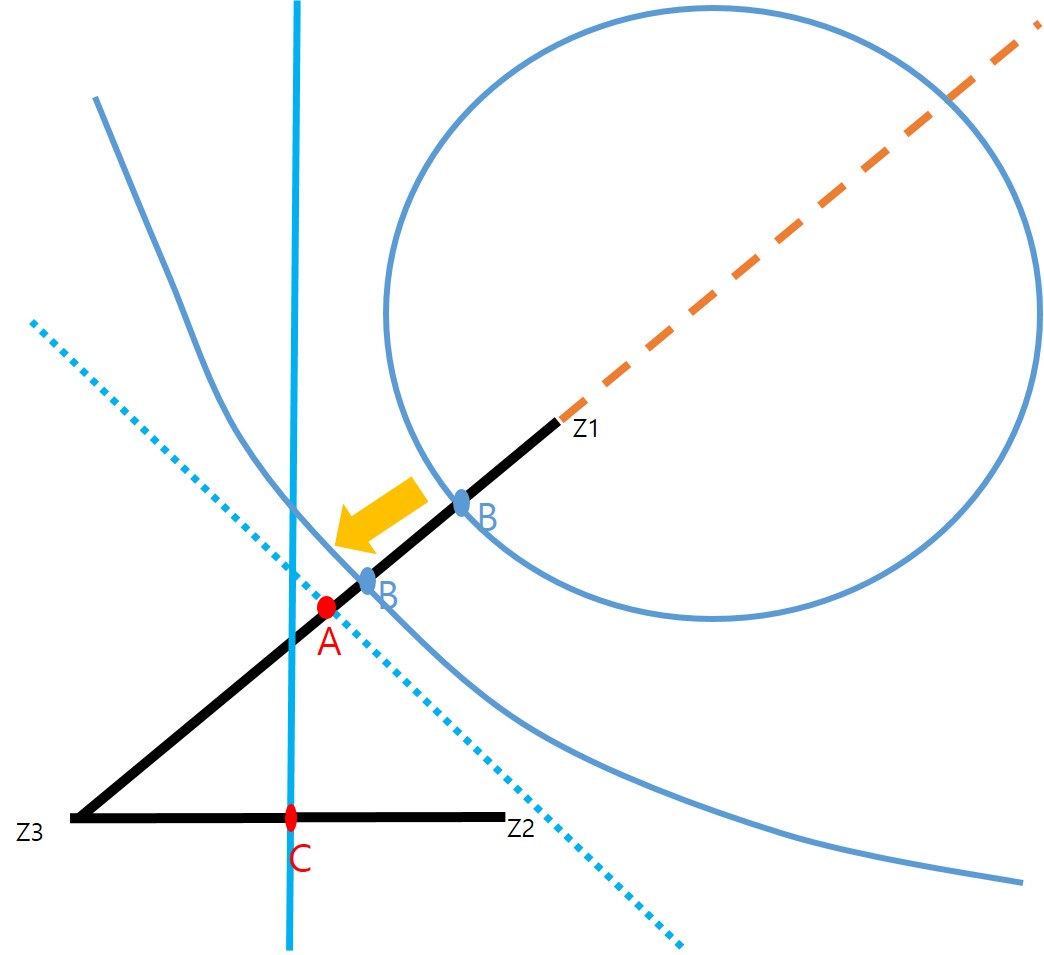,height=4cm,width=4cm,clip=2cm}}
\end{minipage}
\begin{center}
(a)\qquad\qquad\qquad\qquad\qquad\qquad\qquad\qquad\qquad\qquad(b)
\end{center}
\caption{Cases in which the Apollonius circle along $\overline{z_1 z_3}$ does not meet the perpendicular line bisecting $\overline{z_2 z_3}$: (a) Bigger and smaller Apollonius circles along $\overline{z_1 z_3}$. (b) As $d_1$ increases to $d_2$, the Apollonius circle goes to the dotted perpendicular line bisecting $\overline{z_1 z_2}$ and eventually meets the perpendicular line bisecting $\overline{z_2 z_3}$.  }
\label{fig:Dk}
\end{figure}

Next, we consider case (\ref{eq:th1g}). 
\begin{lem}\label{le:th4}
Let  $d_3 > d_2 \ge d_1$ and $|r_{23}-r_{13}| > |b_{23}-b_{13}|$. Then 
\begin{equation}\label{eq:leth4}
r_{23} > r_{13}    \mbox{ and } d_2 > d_1
\end{equation}
\end{lem}
\begin{proof}
Suppose that $r_{13}\ge r_{23}$. Then, by the assumption 
\begin{equation}\label{eq:leth4_1}
r_{13} >  r_{23} + |b_{23}-b_{13}|,
\end{equation}
points $z_3,b_{23}, b_{23}^{in}, b_{13},$ and $b_{13}^{in}$ respectively correspond to $O,A,B,C,D$ in Fig. \ref{fig:Inner} (a). 
In addition, the above equation (\ref{eq:leth4_1}) becomes
\begin{equation}\label{eq:leth4_2}
|DC|>|BA|+|CA|,
\end{equation} 
\begin{eqnarray}\label{eq:leth4_3}
|OD| &=& \frac{d_3}{d_3+d_1}d_{13} = \left(\frac{d_3}{d_1}-1\right)r_{13} 
\nonumber\\
         &\ge& \left(\frac{d_3}{d_2}-1\right)r_{13}  > \left(\frac{d_3}{d_2}-1\right)r_{23}  =  \frac{d_3}{d_3+d_2}d_{23} 
 \nonumber\\ 
         &=& |OB|.
\end{eqnarray}
Therefore, using (\ref{eq:leth4_2}) and (\ref{eq:leth4_3}), we have
\begin{eqnarray*}
|OC|=|OD|+|DC| > |OB| + |BA| + |CA| = |OA| + |CA|,
\end{eqnarray*}
which contradicts the fact that $|OC|,|OA|,$ and $|CA|$ are three sides of a triangle. Therefore, we proved $r_{23}> r_{13}$. \\
Suppose that $d_2=d_1$. In this case, $A=E$ and $C=F$ in Fig. \ref{fig:Inner} (a). Since
\begin{equation*}
r_{23}-r_{13} = \frac{ d_2 d_3 }{d_3^2-d_2^2}(d_{23}-d_{13})  <  \frac{ d_3 ^2}{d_3^2-d_2^2}(d_{23}-d_{13}) 
\end{equation*}
and
\begin{equation*} 
 |b_{23}-b_{13}|=|AC|= \frac{ d_3^2 d_{12} }{d_3^2 - d_2^2}  
\end{equation*} 
we have from the assumption $r_{23}-r_{13}>|b_{23}-b_{13}|$ 
\begin{equation*}
  d_{23}-d_{13} > d_{12},  
\end{equation*}
which contradicts the fact that $d_{12}, d_{13},$ and $,d_{23}$ are the lengths of the sides of a triangle. Thus, we proved $d_2\neq d_1$ and $d_2>d_1$.
\end{proof}
\begin{figure}
\begin{minipage}[t]{6cm}
\centerline{\epsfig{file=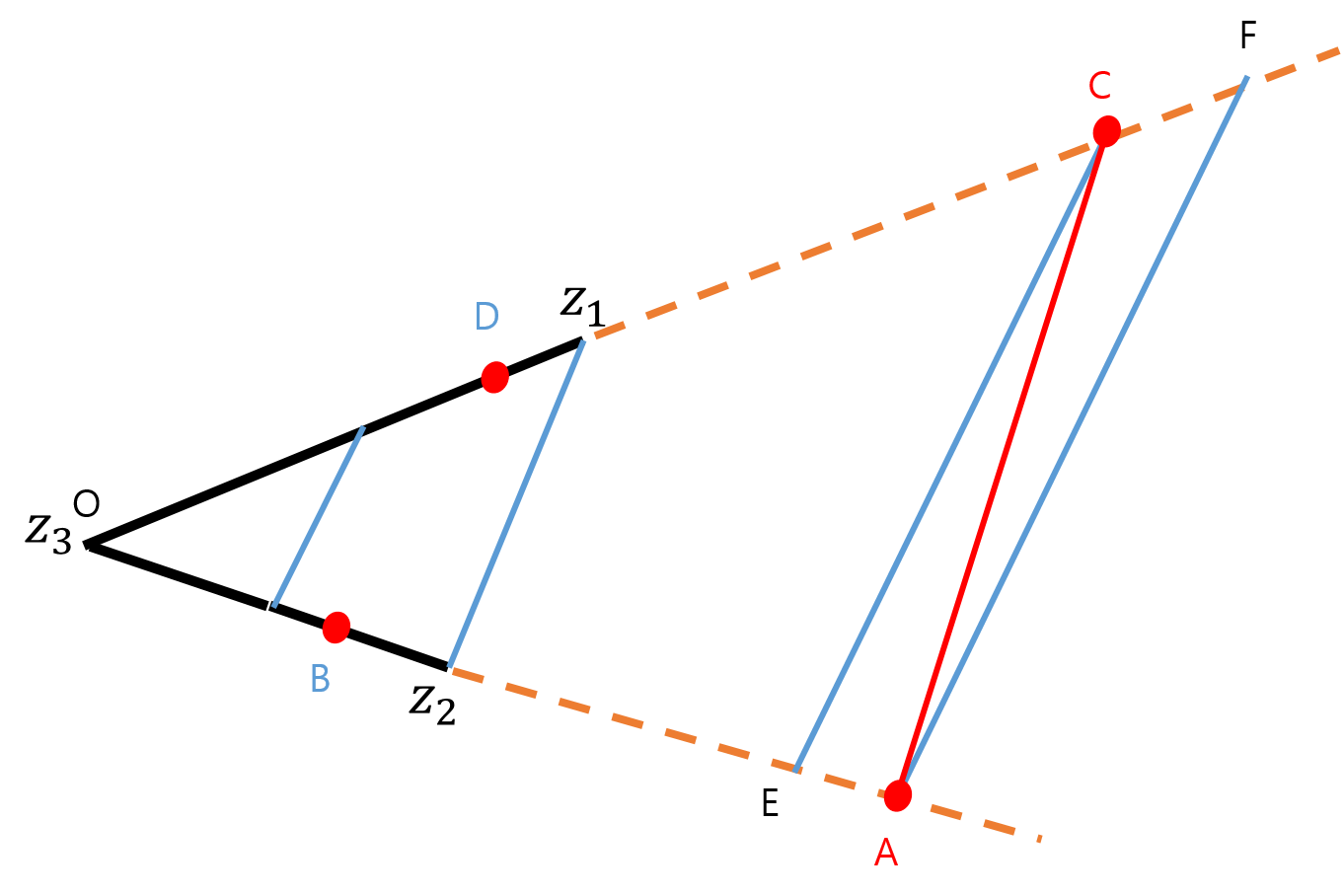, height=4cm,width=4cm,clip=2cm}}
\end{minipage}
\begin{minipage}[t]{4cm}
\centerline{\epsfig{file=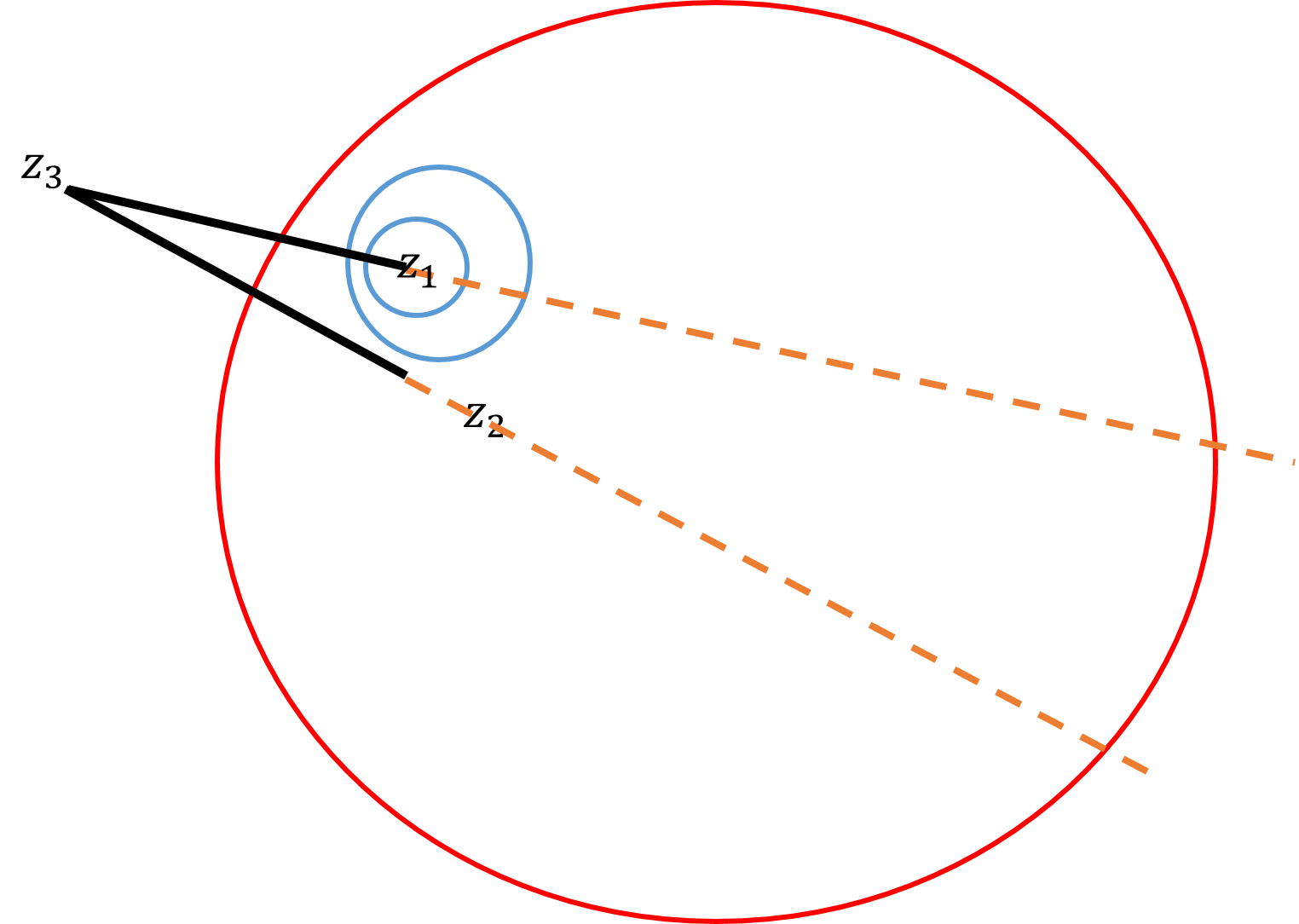, height=4cm,width=4cm,clip=2cm}}
\end{minipage}
\begin{minipage}[t]{6cm}
\centerline{\epsfig{file=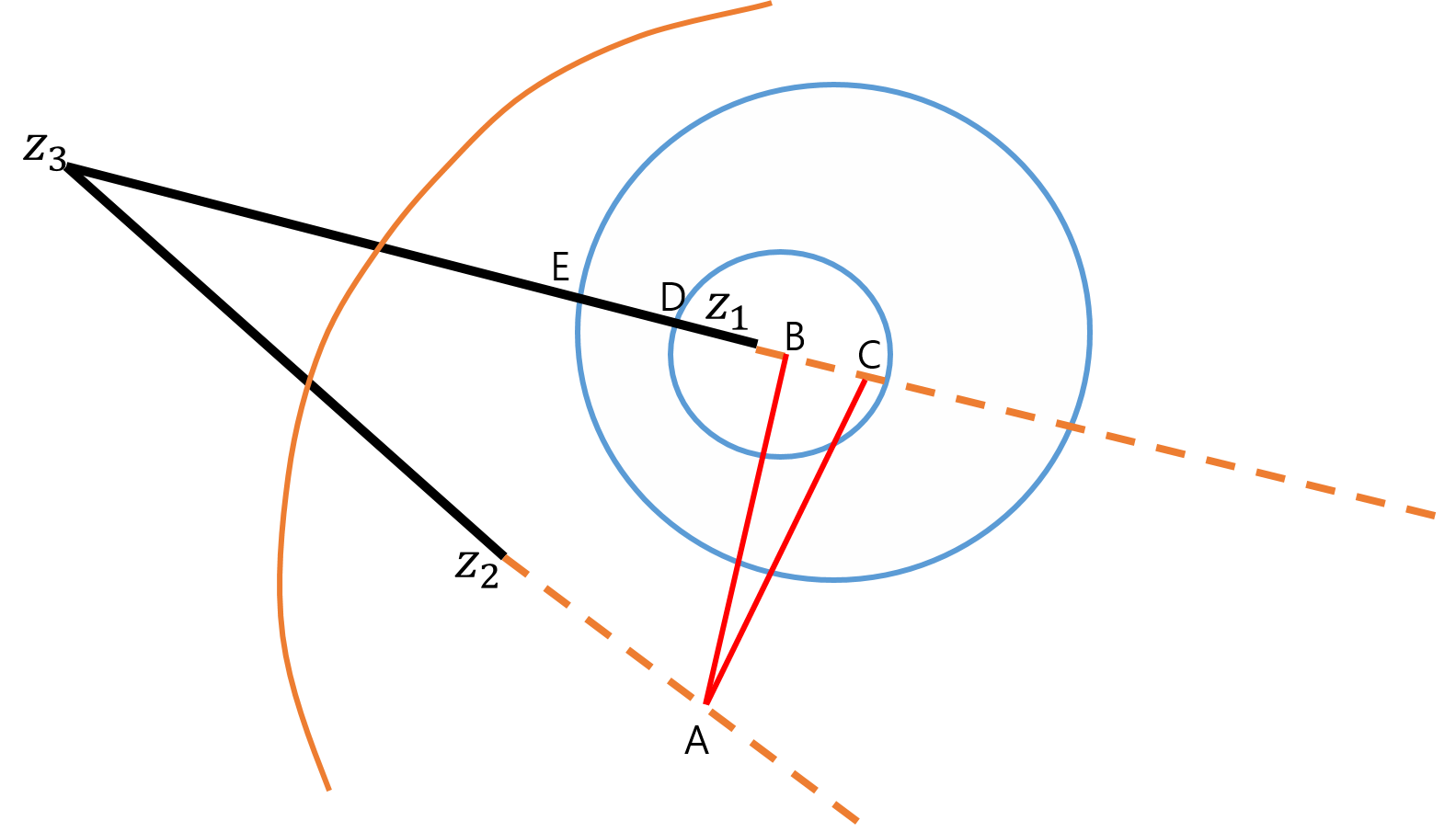, height=4cm,width=4cm,clip=2cm}}
\end{minipage}
\begin{center}
(a)\qquad\qquad\qquad\qquad\qquad\qquad\qquad(b)\qquad\qquad\qquad\qquad\qquad\qquad\qquad(c)
\end{center}
\caption{Cases in which the Apollonius circle along $\overline{z_2 z_3}$ contains the Apollonius circle along $\overline{z_1 z_3}$.
(a) Diagram for Lemma \ref{le:th4}. (b,c) Diagram for Theorem \ref{th:th4} where (c) is an enlargement of (b).}
\label{fig:Inner}
\end{figure}
Under condition (\ref{eq:th1g}) like in Fig.  \ref{fig:Inner} (b), we will increase $d_1$. Let
$$E(k) := ( r_{23} - r_{13}(k,1))- |b_{23}-b_{13}(k,1)|.  $$ 
By Lemma \ref{le:th4}, $E(k)$ is the same as $| r_{23} - r_{13}(k,1)|- |b_{23}-b_{13}(k,1)|$ and the value $k_E := d_2/d_1$ is greater than 1,
\begin{theorem}\label{th:th4}
Suppose $ d_3 > d_2 > d_1$ and $ E(1)>0$.  Then $E(k)$ is a decreasing function for $k\in[1,k_E)$ and  
there is a $k\in[1,k_E)$ such that $E(k)=0$ and $d_3 > d_2 > kd_1$. 
\end{theorem}
\begin{proof}
Suppose that $1<k_1<k_2<k_E$. 
Let $b_{23}, b_{13}(k_1,1), b_{13}(k_2,1), b_{13}^{in}(k_1,1),$ and $b_{13}^{in}(k_2,1)$ be $A,B,C,D,$ and $E$ in Fig. (\ref{fig:Inner}) (c), respectively. Then, using $|AB|,|BC|,$ and $|CA|$'s which are the lengths of the three sides of a triangle $\Delta ABC$, we have 
\begin{eqnarray*}
E(k_2)-E(k_1) &=& -r_{13}(k_2,1)+r_{13}(k_1,1)-|b_{23}-b_{13}(k_2,1)|+|b_{23}-b_{13}(k_1,1)|\\
                           &=&-|CE|+|BD|-|AC|+|AB|
                           =-|ED|-|BC|-|AC|+|AB|\\
                           &=&-|ED|-(|BC|+|CA|-|AB|)
                           <0.
\end{eqnarray*}
Thus, $E$ is a decreasing function for $k\in [1, k_E)$.\\
The Apollonius circle along $\overline{z_2 z_3}$ contains $z_2$ and $z_1$, and therefore also contains the line joining $z_2$ and $z_1$. If $k$ goes to $k_E$ from below, the Aplollonius circle along $\overline{z_1 z_2}$ converges to the perpendicular line bisecting $\overline{z_1 z_2}$. Therefore, two Apollonius circles along $\overline{z_1 z_2}$ and $\overline{z_2 z_3}$ meet for some $k\in[1,k_E)$,  meaning that two Apollonius circles along $\overline{z_1 z_3}$ and $\overline{z_2 z_3}$ also meet at the same value $k$. Therefore, there is some $k\in[1, k_E)$ such that $E(k)=0$.  
\end{proof}

Finally, let us consider case (\ref{eq:th1h}). In this case, $d_1$ and $d_2$ increase at the same rate upto $k_F=\frac{d_3}{d_2}$ times.  Let
$$ F(k_1,k_2) :=   |b_{23}(k_1,1)-b_{13}(k_2,1)| - ( r_{23}(k_1,1) + r_{13}(k_2,1)). \text{ and } F(k) := F(k,k).$$ 
\begin{theorem}\label{th:th3}
Suppose $ d_3 > d_2 \ge d_1$ and $ F(1)>0 $. $F(k)$ is a decreasing function with respect to $k\in[1,k_F)$.
If $d_2=d_1$, there is a $k\in[1,k_F)$ such that $F(k)=0$ and $d_3>kd_2=kd_1$.
\end{theorem}
\begin{proof}
Let  $1< k_1 < k_2 < k_F $.  Note that 
$$ F(k_2) - F(k_1) = [F(k_2)- F(k_1,k_2)]  +  [F(k_1,k_2) - F(k_1)]. $$
We can prove that 
$$ F(k2) - F(k_1,k_2) <0   \mbox{ and } F(k_1,k_2) - F(k_1) < 0 $$ 
in a similar manner to the proofs of Theorem \ref{th:th2} and Theorem \ref{th:th4}.\\
Suppose $d_3>d_2=d_1$. Since $kd_2=kd_1$, the solution should lie on the perpendicular line bisecting $
\overline{z_1 z_2}$. If $k$ goes to $k_F$. the Apollonius circle along $\overline{z_2 z_3}$ goes to the perpendicular line bisectiong $\overline{ z_2 z_3}$. Thus, there must  exist a $k\in[1,k_F)$ such that a perpendicular line bisects $\overline{z_1 z_2}$ and the Apollonius circle along $\overline{z_2 z_3}$.  
\end{proof}
 If $d_2\neq d_1$, there might or might not be a $k\in[1,k_F)$ such that $F(k)=0$. If there is such a $k$, we have $d_3 > kd_2 > kd_1$.
 Otherwise, we should try the case $d_3=d_2'(=k_F d_2)>d_1'(=k_F d_1)$ and increase $d_1'$ upto $d_3$, which is the case in Theorem \ref{th:th2}.
 
To summing up, our strategy for finding $x$ is shown in Fig. \ref{fig:flowchart}.
\begin{figure}
\begin{minipage}[t]{16cm}
\centerline{\epsfig{file=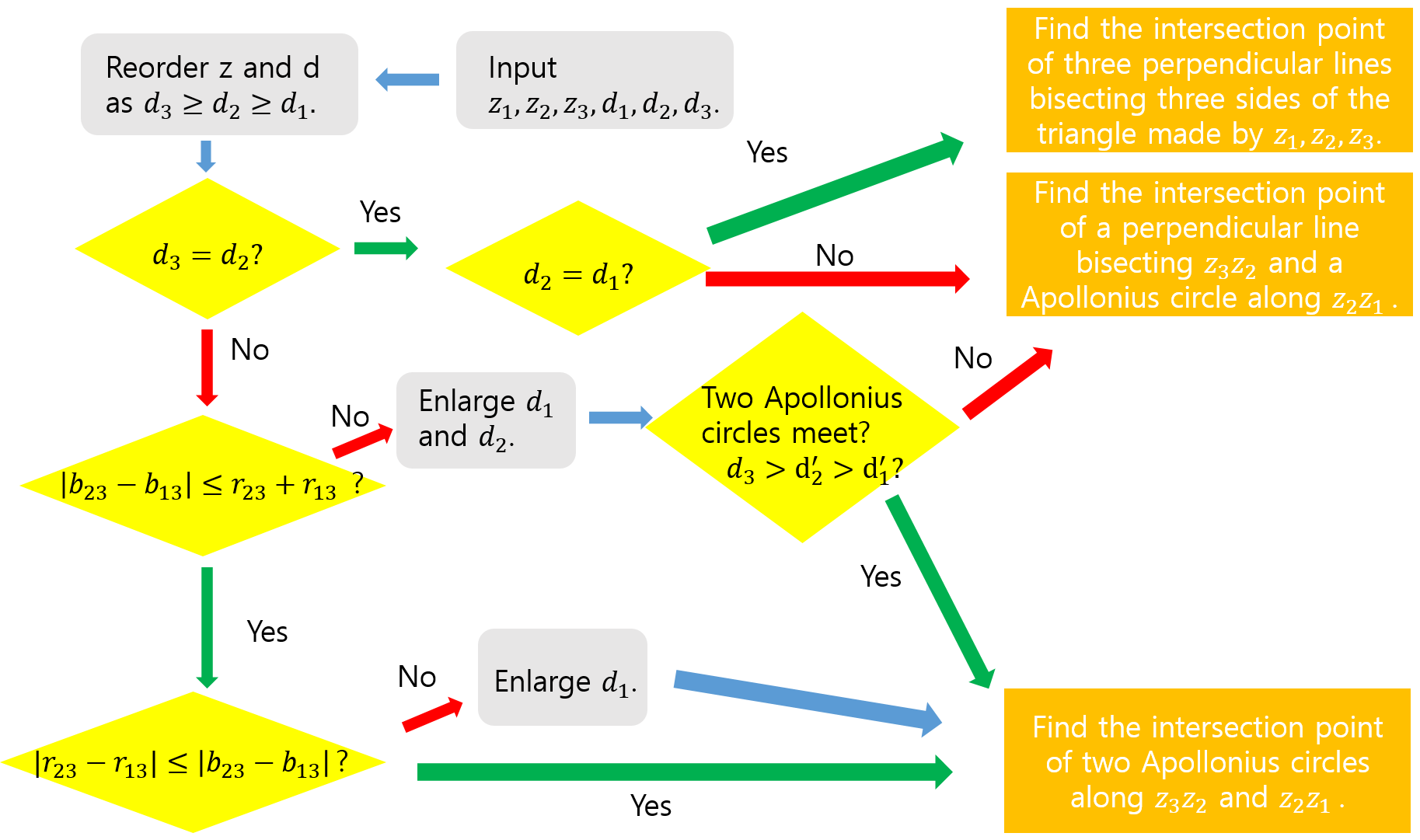, height=10cm,width=16cm,clip=2cm}}
\end{minipage}
\caption{Flow chart of the proposed algorithm}
\label{fig:flowchart}
\end{figure}
We excluded the case in which $d_3=d_2=d_1$ and $z_1,z_2,z_3$ are colinear, as there is no solution to (\ref{eq:main}) in such a case.

\section{Numerical test}
In this section, we numerically tested the proposed algorithm in two ways: numerical illustration of all cases is shown in Fig. \ref{fig:flowchart} and approximation error analysis is used for 50 random samples with respect to the increasing multiplication error. 
 
First, we randomly generated $z_i$ on $(0,1)\times(0,1)$ and $r_i$ on $(0,1)$ for $i=1,2,3$ with uniform probability density, then reordered 
them to satisfy $d_3\ge d_2\ge d_1$. We classified the cases as follows 
\begin{itemize}
\item{Case 003 : $d_3=d_2=d_1$ : There is a solution if $z_1,z_2,$ and $z_3$ are not colinear by (\ref{eq:th1a}).}
\item{Case 013: $d_3=d_2> d_1$ and $\frac{|z_{23}\cdot(b_{13}-b_{23}^{in})|}{d_{23}}\le r_{13} $ : There is a solution  by (\ref{eq:th1b}) and (\ref{eq:th1d}).}
\item{Case 012013: $d_3=d_2> d_1$ and $\frac{|z_{23}\cdot(b_{13}-b_{23}^{in})|}{d_{23}}> r_{13} $: Find a minimum $k$ such that $d_3=d_2>d_1'(=kd_1)$ satisfies a solution. The existence of minimum $k$ is proved in Theorem \ref{th:th2}.}
\item{Case 113: $d_3>d_2\ge d_1$ and $|r_{23}-r_{13}|\le|b_{23}-b_{13}|\le r_{23}+r_{13} $: There is a solution by (\ref{eq:th1c}) and (\ref{eq:th1e}).} 
\item{Case 112+113: $d_3>d_2\ge d_1$ and $|b_{23}-b_{13}|> r_{23}+r_{13} $:  Find a minimum $k$ such that $d_3>d_2'(=kd_2)\ge d_1'(=kd_1)$ satisfies a solution. If $d_2=d_1$, the existence of such a $k$ is proved in Theorem \ref{th:th3}. If there does not exist such a $k$, let $d_3=d_2'> d_1' $ and go to Case 112+013. }
\item{Case 112+013: The condition $d_3>d_2\ge d_1$ and $|b_{23}-b_{13}|> r_{23}+r_{13} $ changed into the condition 
 $d_3=d_2'>d_1'$ and $\frac{|z_{23}'\cdot(b_{13}'-{b_{23}^{in}}')|}{d_{23}'}> r_{13}' $ :  Find a minimum $l$ such that $d_3 = d_2' > ld_1'$ satisfies a solution. The existence of such $l$ is proved in Theorem \ref{th:th2}. }
\item{Case 112-113:  $d_3>d_2\ge d_1$ and $|r_{23}-r_{13}|>|b_{23}-b_{13}| $:   Find a minimum $k$ such that $d_3=d_2>d_1'(=k d_1) $ satisfies a solution. The existence of $k$ is proved in Theorem \ref{th:th4} and Lemma \ref{le:th4}. }
\end{itemize} 

All cases explained above are shown in Fig. \ref{fig:Cases}.  We computed  the multiplicative error ratio $\epsilon$, the normalized range measurement data $d/d(3)$, the normalized distance from the suggested solution $\frac{|z-x|}{|z_3-x|}$, and the Relative Ratio Error(RRE), and the results are displayed In Table \ref{tab:Cases}. RRE is the ratio error (\ref{eq:same}) normalized by $1+\epsilon_3$, defined as follows:
$$  RRE =\frac{ |\epsilon_3-\epsilon_2|+|\epsilon_3-\epsilon_1|}{1+\epsilon_3}
 ,$$
  which measures how close the three values $\epsilon_1,\epsilon_2,$ and $\epsilon_3$ are to each other.

\begin{figure}
\begin{minipage}[t]{6cm}
\centerline{\epsfig{file=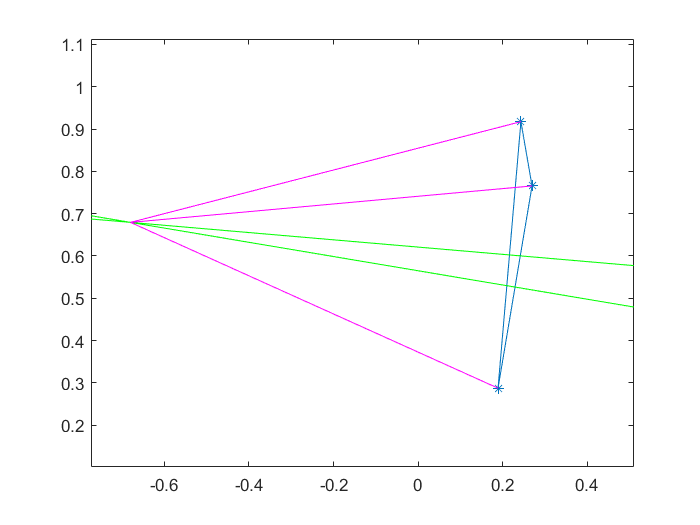, height=4cm,width=4cm,clip=2cm}}
\end{minipage}
\begin{minipage}[t]{4cm}
\centerline{\epsfig{file=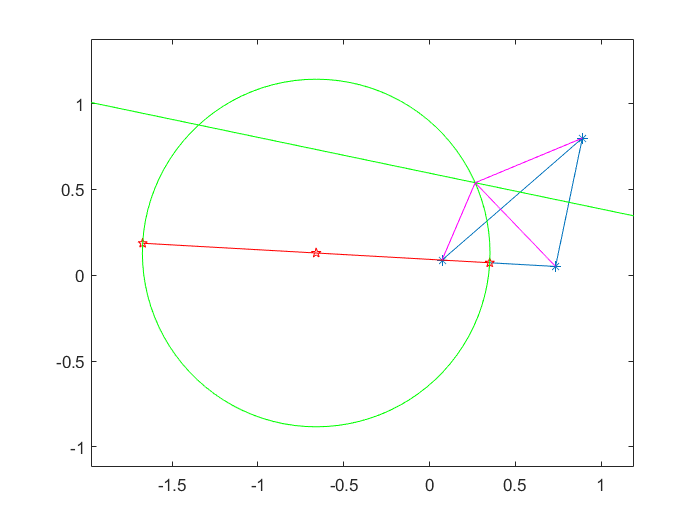, height=4cm,width=4cm,clip=2cm}}
\end{minipage}
\begin{minipage}[t]{6cm}
\centerline{\epsfig{file=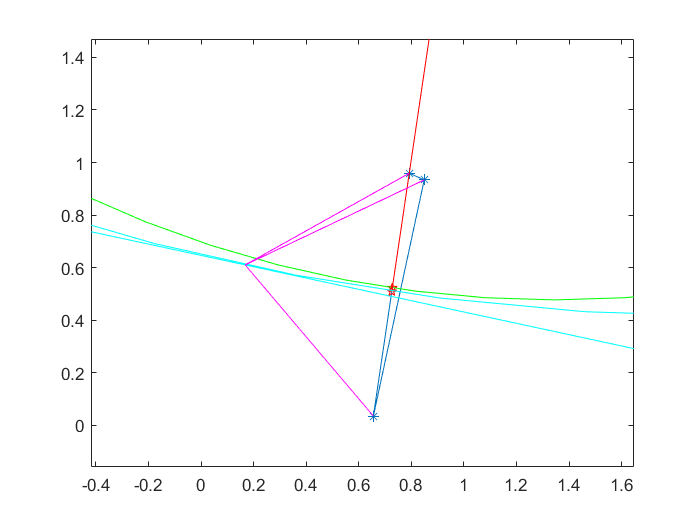, height=4cm,width=4cm,clip=2cm}}
\end{minipage}
\begin{center}
(a)\qquad\qquad\qquad\qquad\qquad\qquad\qquad(b)\qquad\qquad\qquad\qquad\qquad\qquad\qquad(c)
\end{center}
\begin{minipage}[t]{6cm}
\centerline{\epsfig{file=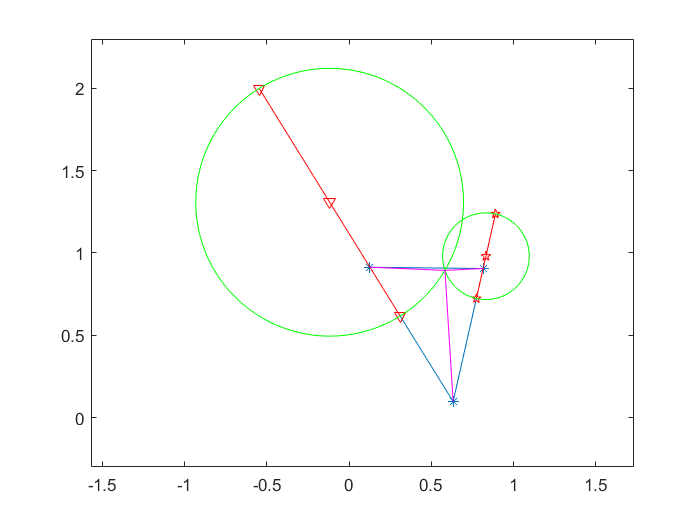, height=4cm,width=4cm,clip=2cm}}
\end{minipage}
\begin{minipage}[t]{4cm}
\centerline{\epsfig{file=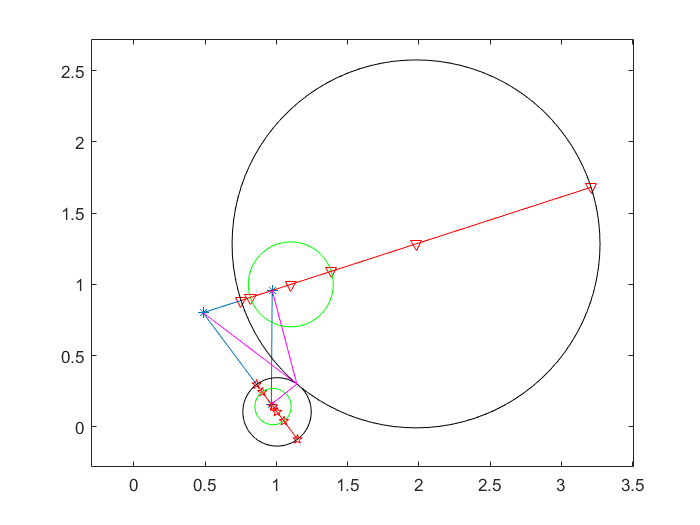, height=4cm,width=4cm,clip=2cm}}
\end{minipage}
\begin{minipage}[t]{6cm}
\centerline{\epsfig{file=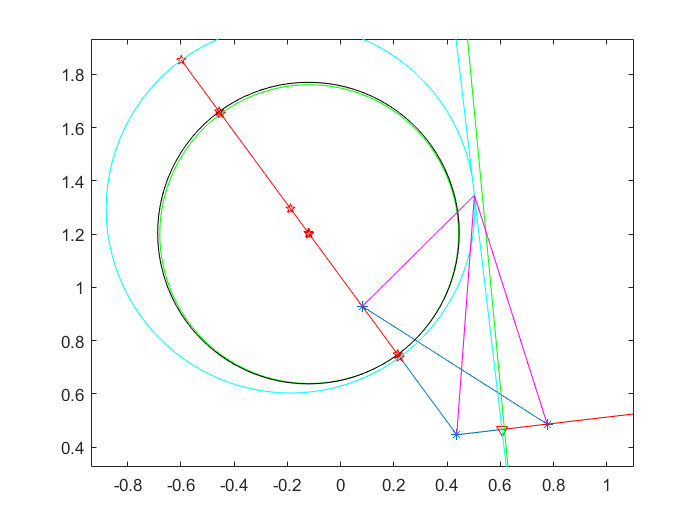, height=4cm,width=4cm,clip=2cm}}
\end{minipage}
\begin{center}
(d)\qquad\qquad\qquad\qquad\qquad\qquad\qquad(e)\qquad\qquad\qquad\qquad\qquad\qquad\qquad(f)
\end{center}
\begin{minipage}[t]{6cm}
\centerline{\epsfig{file=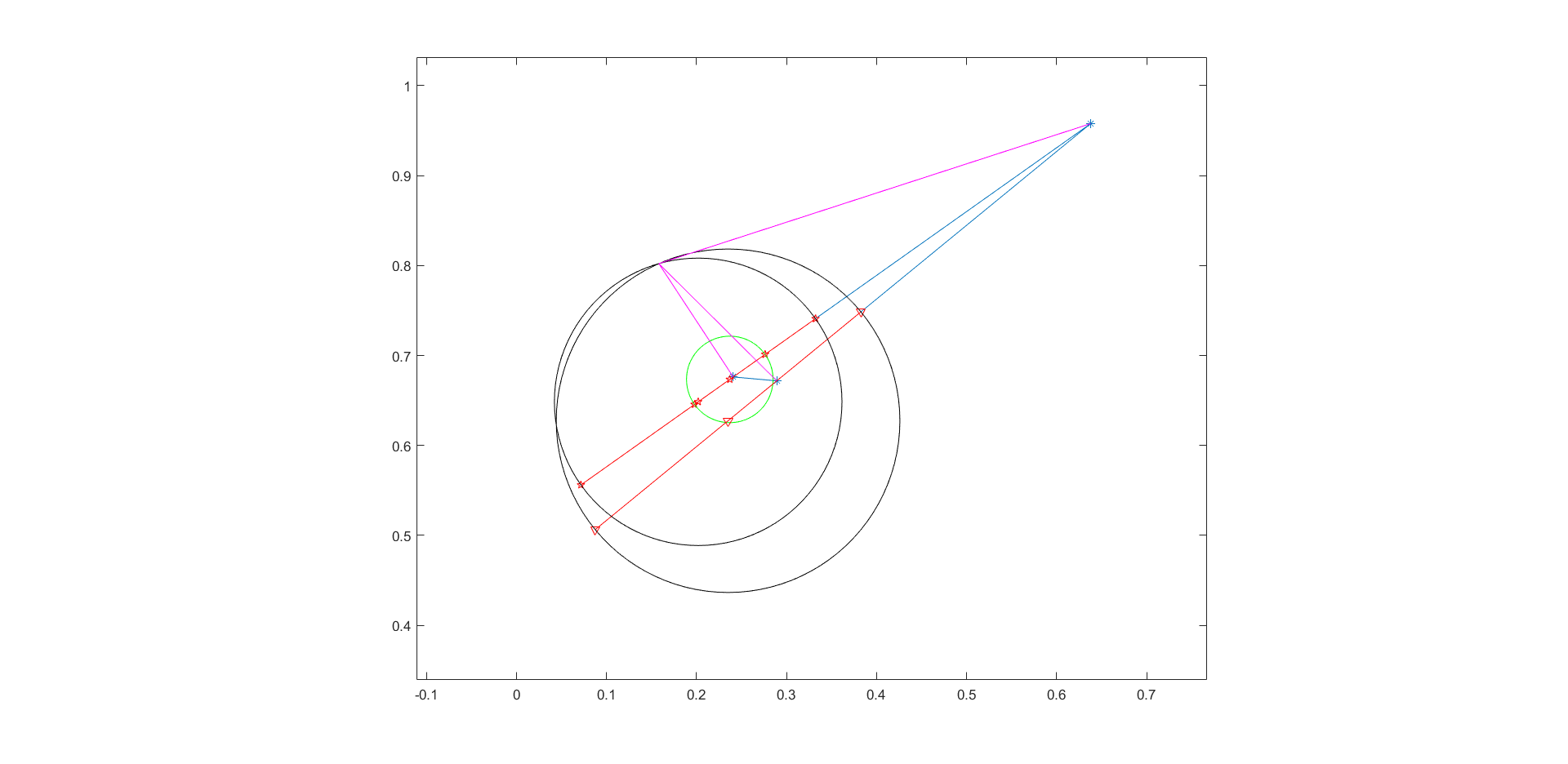, height=4cm,width=7cm,clip=2cm}}
\end{minipage}
\begin{center}
(g)\qquad\qquad\qquad\qquad\qquad\qquad\qquad\qquad\qquad\qquad\qquad\qquad\qquad\qquad
\end{center}
\caption{Process of finding the source location $x$ for call cases  (a)003, (b)013, (c)012013, (d)113, (e)112+113, (f)112+013, and (g)112-113 from Fig. \ref{fig:flowchart}. Magenta lines are the lines from $x$ to the three measurement points $z$.}
\label{fig:Cases}
\end{figure}

\begin{table}
\begin{centering}
\begin{tabular}{|ccccc|}
\hline
Case &   $\epsilon$&  $\frac{d^0}{d_3^0}$& $\frac{|z-x|}{|z_3-x|}$ &RRE   \\
\hline \hline
003 & (-0.2829   -0.2829   -0.2829) &(1.000,1.000,1.0000)& (1.0000,1.0000,1.0000)&1.5483e-16\\
\hline
013 & (0.3943    0.3943    0.3943) &(0.7250    1.0000    1.0000)&0.7250    1.0000    1.0000)&3.1850e-16\\
\hline
012013 &(-0.0509    0.0049    0.0049)&(0.8957,1.0000,1.0000)&(0.9484,1.0000,1.0000)&0.0555\\
\hline
113 &(0.1991,0.1991,0.1991)&(0.2909,0.5712,1.0000)&(0.2909,0.5712,1.0000)&4.4443e-15\\
\hline
112+113 &(-0.3797   -0.3797    0.1067)&(0.1550    0.4606    1.0000)&(0.2765    0.8219    1.0000)&0.8791\\
\hline
112+013 &-0.4829   -0.4360   -0.4335)&(0.5998    0.9956    1.0000)&(0.6570    1.0000    1.0000)&0.0915\\
\hline
112-113 &( -0.5489    0.3784    0.3784)&(0.0978    0.3665    1.0000)&( 0.2989    0.3665    1.0000)&0.6728\\
\hline
\end{tabular}
\caption{Values $\epsilon, \frac{d^0}{d_3^0}, \frac{|z-x|}{|z_3-x|}$ are computed for all cases shown in Fig. \ref{fig:flowchart}. }
\label{tab:Cases}
\end{centering}
\end{table}

To find a solution of (\ref{eq:main}) where $\epsilon$ is given by Gaussian noise, we have computed the mean and standard deviation for  50 samples with multiplicative Gaussian noise from 0\% to 100\%.  The simulated data are obtained throuth the following procedures: First, let us assume that the source location $x^0=(0.5,0.5)$ is fixed. Second, make uniform random measurement points $z_j\in [0,1]\times[0,1], j=1,2,3$. Third, add the multiplicative Gaussian noises with $100*\eta \%, \eta  \in [0, 1]$ to the distances from $x^0$ to $z_j$, resulting in  $d_j = \max(0.01,||z_j-x^0||(1+\eta*G(0,1))$, where $G(0,1)$ is a standard Gaussian noise. Fourth, approximate $x^0$ using the proposed algorithm. Fifth, calculate the mean and standard devation of 50 approximation errors $||x -x^0 ||$, where $x$ is the approximation of $x^0$, for given $\eta$.  The step size for $\eta$ is chosen to be 2\%. The computed errors with respect to $\eta$ are shown in Fig.\ref{fig:Error} (a). Removing the four highest standard deviation points, we can see the approximately increasing behavior shown in Fig. \ref{fig:Error} (b).  

\begin{figure}
\begin{minipage}[t]{8cm}
\centerline{\epsfig{file=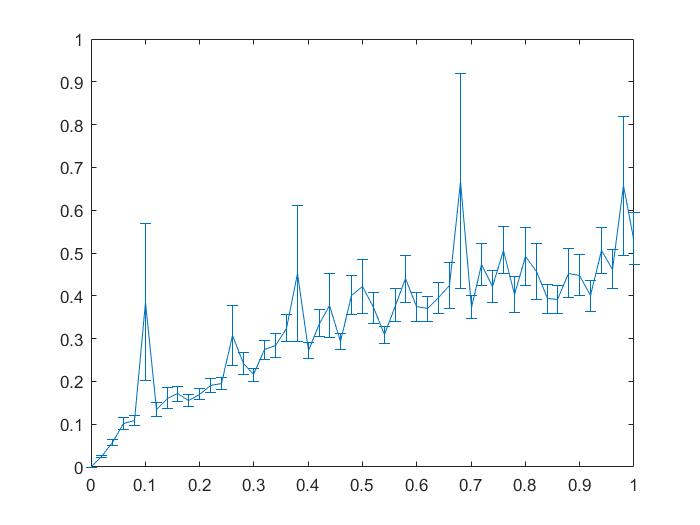, height=4cm,width=6cm,clip=2cm}}
\end{minipage}
\begin{minipage}[t]{8cm}
\centerline{\epsfig{file=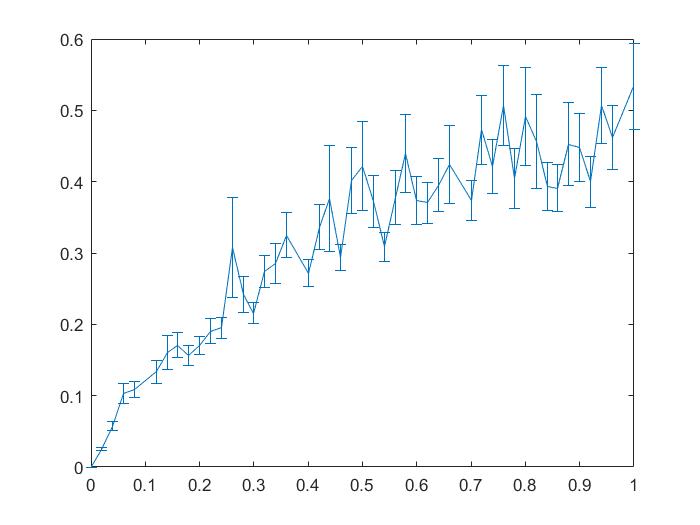, height=4cm,width=6cm,clip=2cm}}
\end{minipage}
\begin{center}
(a)\qquad\qquad\qquad\qquad\qquad\qquad\qquad\qquad\qquad\qquad\qquad(b)
\end{center}
\caption{Means and standard deviations of the errors for the 50 samples used to find the source with respect to multiplicative noise ratio $\eta$ from $0\%$ to $100\%$ using the proposed algorithm. Standard deviation is represented as the vertical bar. (a) Step size is $2\%$. (b) Four points with the highest standard deviation from the graph in (a) are are eliminated. }
\label{fig:Error}
\end{figure}

\section{Conclusions}
In this paper, the theoretical background on locating a singular source from three range measurements with multiplicative noise was exploited. When the multiplicative noise was the same for the points of three measurement data, the equivalent  condition for the existence of the singular source was presented and proved using the idea of Apollonius circles. When there existed solutions, there were one or two. When two solutions existed, we chose the closest point whose distance to $z_3$ was more similar to the longest distance $d_3$ as a possible approximation of the source.  When no solution existed for the same $\epsilon$, we proposed an algorithm with which to find the best approximation with respect to RRE by controlling $\epsilon$. The algorithm preserves the distance length order and minimizes RRE; that is to say, it minimizes the ratio difference among $d_j/|z_j-x|, j=1,2,3$.  Numerical examples for all cases in the algorithm are shown including the measurement triangle, Apollonius circle, and perpendicular line bisecting one of the sides, as well as the approximated solution. Finally, we showed that, as the multiplicative noise ratio increased from $0\%$ to $100\%$, the mean and standard deviations' for the 50 samples increased asymptotically.

\centerline{{\bf Acknowledgements}}
This work was written when the author visited Southern Illinois University Edwardsville, and the author discussed the work with professor Jun Liu. This work is supported by the Basic Research Program through the National Research Foundation of Korea(NRF) funded by the Ministry of Science and ICT (NRF-2017R1A2B4004943).


\end{document}